\documentclass[runningheads]{llncs}
\usepackage{graphicx}

\usepackage{natbib}
\setcitestyle{square,numbers}

\usepackage[toc,page]{appendix}

\usepackage{nameref}

\usepackage{amsmath}
\usepackage{amssymb}
\usepackage{mathtools}
\usepackage[ruled,vlined]{algorithm2e}

\usepackage{xcolor}
\usepackage{color,soul}
\usepackage{pdfpages}

\usepackage{url}

\newcommand{\BibTeX}{\rm B\kern-.05em{\sc i\kern-.025em b}\kern-.08em\TeX}

\usepackage{lipsum}

\newcount\Comments  
\definecolor{darkgreen}{rgb}{0,0.6,0}         
         
\newcommand{\kibitz}[2]{\ifnum\Comments=1{\color{#1}{#2}}\fi}
\newcommand{\rmr}[1]{\kibitz{red}{[Reshef says:#1]}}

\newcommand{\dv}[1]{\kibitz{cyan}{[Dan says:#1]}}

\def\AV{\text{AV}}
\def\CC{\text{CC}}
\def\PAV{\text{PAV}}
\def\sPAV{\text{sPAV}}
\def\RX{\text{RX}}
\def\RXPAV{\text{RX-PAV}}
\def\RXeps{\text{RX-}$\epsilon$}

\begin{document}

\title{Welfare vs. Representation in Participatory Budgeting}
\author{Roy Fairstein\inst{1} \and
Reshef Meir\inst{2} \and Dan Vilenchik\inst{1} \and
Kobi Gal\inst{1,3}}

\institute{Ben-Gurion Univ. of the Negev, Israel \and
Technion, Israel Institute of Technology \and
University of Edinburgh, U.K.}

\maketitle

\begin{abstract}
Participatory budgeting (PB) is a democratic process for allocating funds to projects based on the votes of members of the community.  
Different rules  have been used to aggregate participants' votes.

A recent paper by  Lackner and Skowron~\cite{lackner2020utilitarian}
 studied the trade-off between notions of  social welfare and representation in the multi-winner voting, which is a special case of participatory budgeting with identical project costs.  
 But there is little understanding of this trade-off in the more general PB setting. 
This paper provides a theoretical and empirical study of the worst-case guarantees of several common rules to better understand the trade-off between social welfare and representation.  
We show that many of the guarantees from the multi-winner setting do not generalize to the PB setting, and that  the introduction of costs leads to substantially worse guarantees, thereby exacerbating the welfare-representation trade-off. 
We further study how the requirement of   \emph{proportionality} over voting rules effects the guarantees on social welfare and representation. We study the latter point also empirically, both  on real and synthetic datasets. 
We show  that   variants  of the recently suggested voting rule  Rule-X (which satisfies proportionality) do very well in practice both with respect to social welfare and representation.

\end{abstract}


\section{Introduction}

Participatory budgeting (PB) is gaining attention from both researchers and practitioners  and is actively in use in  cities around the world~\cite{pape2016budgeting, su2017porto, sintomer2008porto}. 
PB includes several steps. First,   a  list is suggested of feasible projects and  their estimated cost. Then,  citizens vote on which of the projects they would like to  fund. Finally, the votes are aggregated by a mechanism (voting rule) that selects a subset of the projects that get funded. The voting rule itself is not hidden from the voters, and they may strategize as they wish.

The voting rule is typically designed to optimize for certain  criteria, the most common are {\em social welfare} (the sum of utilities the voters get from the outcome. In our case, each voter get a utility of 1 for each approve project in the outcome), how many voters are represented in the final outcome, i.e. how many voters got at least one funded project that they voted for (This notation have different names in the literature, in this paper we will call it {\em representation}, same as \citet{lackner2020utilitarian}) and {\em proportionality} (each group in the population is represented in the final outcome according to its size). 
In many instances there is no way to simultaneously guarantee all criteria.\rmr{I don't think social welfare is usually considered as a fairness criterion.}

\paragraph{Running example}
Consider a city with three districts (see Fig.~\ref{fig:distirct}): district~A has 100 citizens, district~B has 90, and district~C has 10. The total budget is $\$1000$ and there are three types of projects: Diamonds (D) that cost $\$200$; Emeralds (E) that cost $\$150$ and Gold (G) that costs $\$100$.
In district~A there are two diamonds and six gold, in district~B there are three diamonds and three emeralds, and in district~C there is one emerald and on gold.
Each citizen approves all the projects  in his district, and no other project. 


The outcome with the optimal social welfare has a value of 800 and  contains all of district A's projects, but represents only 100 citizens which are 50\% of the population. The outcome with the optimal representation requires  a project from each district, but in this case the social welfare cannot exceed 790 (see second line in Table~\ref{tab:example}).
Interestingly, neither of these two outcomes satisfy proportionality, which requires to fund at least five projects from district A and three projects from district B. This proportional outcome have social welfare of 770 and represents of 95\% of the population.


By now, there are various well-known voting rules in the literature, such as  Approval Voting (which maximizes social welfare) or the Chamberlin–Courant  rule \cite{chamberlin1983representative} (which maximizes representation). The example in Figure \ref{fig:distirct} suggests that there is no one-size-fits-all solution, and the rules differ on  the fairness criteria that they guarantee, and their trade-offs. The performance of each rule is estimated both using theoretical analysis (typically a worst-case analysis) \citep{lackner2020utilitarian, skowron2021proportionality}, and data-driven experimental evaluation~\citep{lackner2020utilitarian}\dv{cite such papers}.

\citet{lackner2020utilitarian} studied the trade-off between social welfare and representation in a multi-winner setting, which is equivalent to PB  where all projects have unit costs.  They establish guarantees on the social welfare and representation for 12 voting rules from the literature.
In most of these results, the derivation of the bounds relies heavily on the assumption of identical costs, and hence do not readily extend to the general PB setting (or not at all, as we show for some rules). \rmr{no reason to jump ahead of yourself - you describe the main results in the next subsection not here}

Furthermore, it is of interest to understand the ``cost of proportionality'': Rather than ad-hoc analysis of specific voting rules that happen to satisfy proportionality, we would like to understand what is the inherent tradeoff in social welfare (or representation) that we must pay by requiring proportionality. Furthermore, some rules that satisfy proportionality in multi-winner setting do not satisfy it in the PB setting, stressing  the need for a general analysis.

\begin{figure}[t]
\begin{center}
\includegraphics[width=8cm]{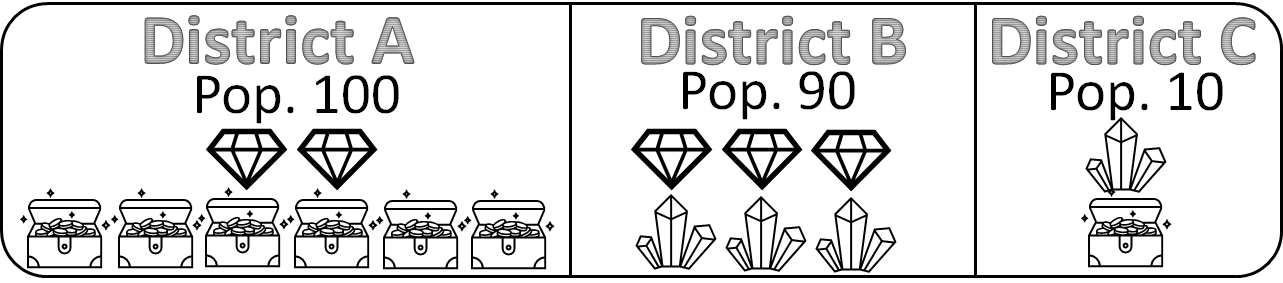}
\caption{Running Example of a Participatory Budgeting (PB) Instance. Three districts with three types of projects: diamonds (cost \$200), emeralds (\$150) and gold (\$100).
}\label{fig:distirct}
\end{center}\vspace{-5mm}
\end{figure}


    
\subsection{Our Contribution}
In this paper, we  extend the theoretical guarantees of \citet{lackner2020utilitarian} from the multi-winner setting to  participatory budgeting (PB)  and analyze the  trade-off between social welfare and representation for popular rules from the literature. In addition, we derive tight guarantees for a class of rules called \emph{proportional rules}. 
Those   guarantee a different notion of fairness than the one studied in \cite{lackner2020utilitarian}. The impatient reader can skip directly to Table~\ref{tab:guarantees_util} and Table~\ref{tab:guarantees_rep} to see a summary of our theoretical results.

 
Beyond the theoretical contribution described above, we are   first to compare the welfare-fairness tradeoff of several popular voting rules from the literature on real PB instances; 
We also evaluate the voting rules on two synthetically generated datasets, that allow us to demonstrate the intricate relationship between the different fairness criteria.

From  our theoretical and empirical results, we conclude that PAV continues to provide a good tradeoff between welfare and representation in the presence of project costs. But, while in practice PAV returns a proportional outcome for many instances, it does not guarantee to always do so. On the other hand, we show that sequential PAV becomes substantially worse. In addition, we propose two variants of Rule~X~\citep{peters2020proportional}, whose asymptotic guarantees of welfare and fairness are worse than PAV's, but in practice do just as well on average, while also guaranteeing  a proportional outcome.

Our results provide PB organizers with explainable recommendations on what voting rules are suggested to use, depending on the criteria they care about. 





\section{Related Work}\label{sec:related}

Multiple papers in the participatory budgeting literature focus on either social welfare, representation or proportionality. For example, \citet{goel2019knapsack} suggest using knapsack voting in order to improve the outcome social welfare, and
 \citet{jain2020participatory} consider special cases where it is possible to find a polynomial time algorithm which maximizes the social welfare. \citet{skowron2020participatory} suggest new PB voting rules and empirically evaluates their social welfare and representation. \rmr{any of these rules is in here? if not why not?}
 
 As for proportionality, there are many papers who deal with the subject, suggesting different definitions~\cite{peters2020proportional, fain2016core, fain2018fair, aziz2018proportionally, aziz2017justified, sanchez2017proportional, skowron2020participatory}. In this paper we will focus on specific definition for proportionality called Extended Justified Representation (EJR), which was defined by \citet{peters2020proportional} for PB, as EJR is both a strong requirement, and one that can be guaranteed. 
 
\citet{michorzewski2020price} consider  the trade-off between social-welfare and proportionality in divisible participatory budgeting, i.e., where it possible to fund parts of projects, instead of only entire projects in our case. \rmr{well this sounds very relevant! what do they say? what is the relation to your work?} 
\citet{skowron2021proportionality} analyzed the trade-off between welfare and proportionality in the multi-winner setting, showing for different voting rules the minimal welfare each cohesive group of voters is guaranteed.
\rmr{'talk about' is not an academic context. do they analyze it? run simulations? what are the insights? are they related to your work?} 

The purpose in this paper is to consider all three measurements at once in the PB context. Even though combining \cite{skowron2021proportionality,lackner2020utilitarian} provides a comparison between the three measurements for some rules, it is done in the multi-winner setting  where projects have a unit cost. Introducing different costs can have a significant effect on the results of voting rules which don't take it into account, thus affecting their guarantees. In addition, while some of the voting rules are guaranteed to give a proportional outcome in the multi-winner setting, this isn't correct anymore for PB. 


\section{Preliminaries}\label{sec:prem}
For any  $a\in \mathbb{N}$,  we use $[a]$ to denote $\{1,\ldots,a\}$.
A PB instance is a tuple $E=(A,cost,L)$, where:
\begin{itemize}
    \item Given a set $P\coloneqq\{p_1,\ldots,p_m\}$ of candidate projects, and $V=[n]$ a set of voters, the approval profile $A:V \to 2^P$ maps voter $i \in V$ to $A(i)$, the set of projects that voter $i$ approves.

    \item The mapping  $cost: P\rightarrow \mathbb{R}_+$ assigns a cost to every $p\in P$. The cost for a subset $T \subseteq P$ satisfies $cost(T)=\sum_{p\in T}cost(p)$.
    \item $L\in \mathbb{R}_+$ is the total budget. 
\end{itemize}
 Denote by $c_{\min}:=\min_{p\in P}cost(p);~c_{\max}:=\max_{p\in P}cost(p)$.  We say that $E$ is a multi-winner (MW) problem if $c_{min}=c_{max}=1$. 
 In addition we will use the following notation:
 \begin{itemize}
     \item $\mathcal{E}(N,M,L,c_{min},c_{max})$ is the set of all possible PB instances with $N$ voters, $M$ projects, budget $L$ and minimum and maximum project costs $c_{min},c_{max}$. The set $\mathcal{E} = \bigcup \mathcal{E}(N,M,L,c_{min},c_{max})$ is the union over all possible values of $N,M,L,c_{min},c_{max}$.
     
     \item A  bundle $B\subseteq P$ of projects  is {\em feasible} if  $cost(B)\leq L$. Given PB instance $E$, $S(E)$ is the set of feasible bundles w.r.t. $E$.
\item A {\em voting rule} is a function, such that $\forall E\in\mathcal{E}, R(E)\subseteq S(E)$. This function maps a PB instance $E$ to a set of feasible bundles, referred to as the
\emph{outcome} of the voting rule. 
\end{itemize}
Finally, we measure the outcome with two common metrics:
\begin{itemize}
\item The {\em social welfare} score of an approval profile $A$ with respect to a bundle $B$ is 
$SW(A,B)=\sum_{i\in V}|A(i)\cap B|$
\item The {\em representation } score of an approval profile $A$ with respect to a bundle $B$ is $RP(A,B)=\sum_{i\in V}min(1,|A(i)\cap B|)$
\end{itemize}

\subsection{Popular PB Voting Rules}\label{sec:vr}
We turn to describe several popular voting rules from the literature that we later analyze. 

\textbf{Approval Voting (\AV{})} The rule selects a feasible bundle $B\subseteq P$ that maximizes the social welfare $SW(A,B)$.

\textbf{Approval Chamberlin–Courant (\CC{})\cite{chamberlin1983representative}} The rule selects a feasible bundle $B\subseteq P$ that maximize 
$RP(A,B)$. 

The following rule change \AV{} such that the score achieved from a voter for projects decrease as more of his approved projects are funded. This way increasing the amount of voters that get represented.
\textbf{Proportional Approval Voting (\PAV{})\cite{orsted1894oversigt}} The rule selects a feasible bundle $B\subseteq P$ that maximizes the following score:\\
$$SC_{\PAV{}}(A,B)=\sum_{i\in V}\sum_{k=1}^{|A(i)\cap B|}\frac{1}{k}.$$


\textbf{Sequential-\PAV{} (\sPAV{})} Solving \PAV{} is NP-hard \cite{aziz2014computational}; Sequential \PAV{} is an efficient heuristic that proceeds as follows. Start with an empty bundle $B_0=\emptyset$; in iteration $i$ select a project $p\in P$, among all projects such that $B_{i-1} \cup \{p\}$ is feasible, that maximizes $SC_{\PAV{}}(A,B_{i-1} \cup \{p\})$. Set $B_{i}=B_{i-1} \cup \{p\}$. Repeat until no project $p$ can be added.

\medskip
{\em Remark:} The outcome of a voting rule may contain several optimal bundles and thus require some tie-breaking rule. We specify the approriate rule when needed.


 \begin{table}[t]
  \begin{center}
    \begin{tabular}{|l|c|c|c|}
    \hline
      & District~A & District~B & District~C\\
      \hline
      \AV{} & 2D , 6G & - & -\\
      \CC{} & 1D & 3D & 1E , 1G\\
      \PAV{} & 5G & 3E & -\\
      \sPAV{} & 5G & 3E & -\\
      Any EJR rule & 5G & 3E & -\\
      \hline
    \end{tabular}
    \caption{\label{tab:outcomes}Outcomes for the PB running example (D, E and G for diamond, emerald and gold respectively).}
  \end{center}
\end{table}

Table~\ref{tab:outcomes} shows the outcome of each voting rule when applied to the  running example. As shown by the table, the rules vary widely in their outcome. 
For example,  the voting rule \AV{} does not choose any of the projects in District $B$ and $C$, while $\CC{}$ chooses projects from all districts. 

\subsection{Proportional Voting Rules}\label{sec:propvr}

 
Proportional voting rules ensure  that sufficiently large groups of voters that share  a large set of approved projects must also receive a  fair amount of  projects in the outcome. The key is the notion of $T$-cohesive groups which  are  groups of voters that share a subset of projects $T$ 
and are able to  fund $T$ with the proportional part of the budget. Such groups are `entitled' to a fair representation in the outcome of the PB instance. Formally:


\begin{definition}[T-cohesive group~\cite{peters2020proportionality, peters2020proportional}]\label{def:cohis} A group of voters $S\subseteq V$ that jointly approves a set of projects $T\subseteq \cap_{i\in S}A(i)$ is $T$-cohesive if $\frac{L}{|V|}|S|\geq cost(T)$. 
\end{definition}\label{def:tco}

\begin{definition}[Extended Proportionality Representation (EJR)~\cite{peters2020proportionality, peters2020proportional}\footnote{Not to confuse with representation}]\label{def:EJR} 

A bundle $B$ for PB instance $E=(A,cost,L)$ satisfies EJR if for every  $T\subseteq P$ and every $T$-cohesive group $S$, it holds that there is $i\in S $ such that $|A(i)\cap B|\geq |T|$.
A voting rule $R$ satisfies EJR, if for every PB instance $E$, every bundle in the outcome $R(E)$ satisfies EJR.

\end{definition}

In our running example,  district~A voters  are $T$-cohesive for any set $T$ of 5 cheap projects (gold) in district~A, and district~B voters are $T$-cohesive for the set  $T$ of 3 cheap projects in district~B.  Any  voting rule that satisfies EJR must include 5 projects approved by district~A voters, and 3 projects approved by    district~B voters. As can be seen in Table~\ref{tab:outcomes}, both \PAV{} and \sPAV{} satisfy EJR on this example, while \AV{} and \CC{} do not. In general, none of the rules in Section \ref{sec:vr} are guaranteed  to satisfy EJR (as shown in the running example for \AV{} and CC, and for \PAV{}, \sPAV{} by \citet{peters2020proportional}).

 A well-known rule from the literature that satisfies the EJR  property was  suggested by Peters et al.~\citep{peters2020proportional}, and is called Rule~X (\RX{} for short). This voting will be used later in Section~\ref{sec:exp} as a representative of the EJR voting rules. Rule~X is not as simple to describe as the aforementioned rules, therefore, it will be described in detail in Appendix~\ref{app:rx}
 

\medskip
{\em Remark.} The EJR property does not require that a voting rule exhaust the entire budget. Any voting rule that does not exhaust the budget cannot achieve an optimal social welfare, as adding any project (that at least one voter approved) with the leftover budget will increase the welfare (and possibly the representation). There are many ways to make sure the voting rule exhaust the budget, e.g. \citet{peters2020proportional} do so by giving some very small gain to projects that voters did not approve, this way making sure that \RX{} outcome use the entire budget.


\subsection{Worst-Case Guarantees}\label{sec:wc}

We follow the definitions of \citet{lackner2020utilitarian} for utilitarian and representation guarantees. 
Given a participatory budgeting instance $E$, the
\emph{utilitarian ratio} of a voting rule $R$  for instance $E$ is the proportion of the social welfare given by $R$ (in this case, ties are broken according to the minimum social welfare over all bundles $B$ in the outcome of $R(E)$) divided by the optimal social welfare over all feasible bundles, the set $S(E)$.

\begin{equation}\label{eq:util_ratio}
K_{SW}^R(E) \coloneqq \frac{SW(A,R(E))}{\max_{B\in S(E)}SW(A,B)}
\end{equation}


The  worst-case \emph{utilitarian guarantee} of rule $R$ is  the minimal utilitarian ratio:
\begin{equation}K_{SW}^R(N,M,L,c_{min},c_{max}) \coloneqq \underset{E\in\mathcal{E}(N,M,L,c_{min},c_{max})}{\inf}K_{SW}^R(E)
\end{equation}

In the same way $K_{RP}^R(N,M,L,c_{min},c_{max})$ is the worst-case {\em representation guarantee} of rule $R$; this time ties are broken according to the bundle with the worse representation in $R(E)$.

When omitting one or more of the arguments $N,M,L,c_{min},c_{max}$ in $K_{RP}$ or $K_{SW}$, we are taking the infimum over these arguments. E.g. $K^R_{SW}(N,c_{min}):=\inf_{L,M,c_{max}}K^R_{SW}(N,M,L,c_{min},c_{max})$.



\vspace{1mm}

Table~\ref{tab:example} shows the social welfare and representation ratios for our running example. By definition, the utilitarian guarantee of \AV{} and representation guarantee of CC equal to 1.
 
 \begin{table}[t]
  \begin{center}
    \begin{tabular}{|l|c|c|c|c|c|}
    \hline
      & \AV{}  & CC & \PAV{} & \sPAV{} & Any EJR voting rule\\
      \hline
      SW & 800 & 390 & 770 & 770 & 770 \\
      RP & 100 & 200 & 190 & 190 & 190 \\
      \hline
      $K_{SW}^R(E)$ & 1 &	0.4875 &	0.9625 &	0.9625 &	0.9625\\
      $K_{RP}^R(E)$ & 0.5 &	1 &	0.95 &	0.95 &	0.95\\
      \hline
    \end{tabular}
    \caption{\label{tab:example}The welfare~/~representation scores and their ratios achieved by different voting rules for the running  example (ratios are defined in Section \ref{sec:wc})} 
  \end{center}
\end{table}


\section{Worst-case Guarantees of PB Voting Rules}\label{sec:guar}

\begin{table*}[ht!]
  \begin{center}
    \begin{tabular}{|c|cc|cc|}
    \hline
    &\multicolumn{2}{c|}{Participatory budgeting}& \multicolumn{2}{c|}{Multi-winner}\\
     Rule & Lower & Upper & Lower & Upper\\
      \hline
      & & \\[\dimexpr-\normalbaselineskip+3pt]
     \AV{} & 1 & 1 & 1& 1\\[0.05cm]

    CC & $\Omega\left(\frac{1}{L}\right)$ (Prop.~\ref{prop:cc_low})& $O\left(\frac{1}{L}\right)$ (Prop.~\ref{prop:cc_high}) & $\Omega\left(\frac{1}{L}\right)$ & $O\left(\frac{1}{L}\right)$  \\[0.2cm]

     \PAV{} & $\Omega\left(\frac{\log(L)}{L}\right)$  (Prop.~\ref{prop:pav_util_low})& $O\left(\frac{\log(L)}{L}\right)$ (Prop.~\ref{prop:pav_util_high})& $\Omega\left(\frac{1}{\sqrt{L}}\right)$   & $O\left(\frac{1}{\sqrt{L}}\right)$ \\[0.2cm]

     \sPAV{} & $\Omega\left(\frac{1}{NL}\right)$ (Prop.~\ref{prop:seqpav})& $O\left(\frac{1}{N}\right)$ (Prop.~\ref{prop:seqpav})& $\Omega\left(\frac{1}{\sqrt{L}}\right)$   & $O\left(\frac{1}{\sqrt{L}}\right)$ \\[0.2cm]

\hline\hline     
& & \\[\dimexpr-\normalbaselineskip+3pt]
          EJR rules & $\Omega\left(\frac{1}{Nc_{max}}\right)$ (Prop.~\ref{prop:ejr_util_low})& $O\left(\frac{1}{\sqrt{N}}\right)$ (Prop.~\ref{prop:ejr_util_high})& $\Omega\left(\frac{1}{N}\right)$(Prop.~\ref{prop:mw_util_low})& $O\left(\frac{1}{\sqrt{L}}\right)$ \\[0.2cm]

     \hline
    \end{tabular}
  \caption{Utilitarian guarantees for PB and multi-winner for rules studied in the paper, as a function of the budget ($L$), the number of voters ($N$) and the highest project cost $c_{max}$. We assume w.l.o.g. that the cost of the cheapest project is $1$. \rmr{You can replace the \cite{lackner2020utilitarian} with a \# and explain in the caption, to avoid something like $\frac12 [12]$. Also, the \AV{} bound does not need a reference?}
  All of the multi-winner guarantees are taken from \citet{lackner2020utilitarian}.
  }\label{tab:guarantees_util}
  \end{center}
\end{table*}

\begin{table*}[ht!]
  \begin{center}
    \begin{tabular}{|c|cc|cc|}
    \hline
    &\multicolumn{2}{c|}{Participatory budgeting}& \multicolumn{2}{c|}{Multi-winner}\\
     Rule & Lower & Upper & Lower & Upper\\
      \hline
      & & \\[\dimexpr-\normalbaselineskip+3pt]
       \AV{} & $\Omega\left(\frac{1}{Lc_{max}}\right)$ (Prop.~\ref{prop:av_low})& $O\left(\frac{1}{L}\right)$ (Prop.~\ref{prop:av_high}) & $\Omega\left(\frac{1}{L}\right)$  & $O\left(\frac{1}{L}\right)$ \\[0.15cm]
     
     \CC{} & 1 & 1 & 1 & 1\\[0.1cm]
     
     \PAV{} &  $\Omega\left(\frac{1}{\log(L)}\right)$ (Prop.~\ref{prop:pav_rep_low})  & $O\left(\frac{1}{\log(L)}\right)$ (Prop.~\ref{prop:pav_rep_high}) & $\frac{1}{2}$  & $\frac{1}{2} + O\left(\frac{1}{L}\right)$ \\[0.2cm]

     \sPAV{} &  $\Omega\left(\frac{1}{N}\right)$  (Prop.~\ref{prop:seqpav}) & $O\left(\frac{1}{N}\right)$ (Prop.~\ref{prop:seqpav}) & $\Omega\left(\frac{1}{\log(L)}\right)$   & $\frac{1}{2} + O\left(\frac{1}{L}\right)$ \\[0.2cm]
     
  \hline\hline  
  & & \\[\dimexpr-\normalbaselineskip+3pt]
     EJR rules &  $\Omega(\frac{1}{N})$ (Prop.~\ref{prop:ejr_rep}) & $O(\frac{1}{N})$ (Prop.~\ref{prop:ejr_rep}) & $\Omega\left(\frac{1}{N}\right)$(Prop.~\ref{prop:mw_rep_low})& $\frac{3}{4} + O\left(\frac{1}{L}\right)$ \\[0.2cm]

     \hline
    \end{tabular}
  \caption{Representation (bottom) guarantees for PB and multi-winner for rules studied in the paper, as a function of the budget ($L$), the number of voters ($N$) and the highest project cost $c_{max}$. We assume w.l.o.g. that the cost of the cheapest project is $1$. The multi-winner guarantees are taken from \citet{lackner2020utilitarian}.
  }\label{tab:guarantees_rep}
  \end{center}
\end{table*}

In this section, we describe our first contribution of computing  the worst case welfare and representation guarantees for voting rules from Section \ref{sec:prem}.  We then compute the worst case guarantees for the family of rules that satisfy the EJR property.
Table~\ref{tab:guarantees_util} and Table~\ref{tab:guarantees_rep} show a summary of all our theoretical guarantees, side-by-side with the results of Lackner and Skowron~\cite{lackner2020utilitarian} for multiwinner voting. 

This section will feature the lower bound on social welfare guarantee of \PAV{} and the lower and upper bounds on \sPAV{}. The proofs of the other bounds and rules are left for the appendix. The proofs in the main body of the paper represent the spirit of how we derive lower bounds (general argumentation) and upper bounds (construction of a certain PB).

Notice that proofs from \citet{lackner2020utilitarian} for the MW setting rely heavily on the fact that costs are uniform, which fail in the PB setting.There are a few proofs that follow the same outlines as \citet{lackner2020utilitarian} and we shall point this out.

\subsection{Common Voting Rules}


We start with the guarantees for \AV{}, \CC{}, \PAV{}  and \sPAV{}.

\begin{proposition}\label{prop:pav_util_low}
$\forall L,c_{min}:\ K^{\PAV{}}_{SW}(L,c_{min})\geq \frac{c_{min}}{L}\log\left(\frac{L}{c_{min}}\right)$
\end{proposition}

The heart of the proof lies in the following technical lemma, whose proof is deferred to right after the proof of this claim.

\begin{lemma}\label{lemma:pav}
For any feasible bundle $B$ of projects and any approval profile $A$ it holds $\frac{SC_{\PAV{}}(A,B)}{SW(A,B)}\geq \frac{c_{min}}{L}\log\left(\frac{L}{c_{min}}\right)$.
\end{lemma}
\begin{proof}[Proof of Prop.~\ref{prop:pav_util_low}]
Given a PB instance $E$, we denote by $B_{SW}$ the bundle with largest SW, and $B_{\PAV{}}$ the one with largest \PAV{} score. From Lemma~\ref{lemma:pav} the following holds:

\begin{align*}
SW(A,B_{\PAV{}}) &\geq SC_{\PAV{}}(A,B_{\PAV{}}) \geq SC_{\PAV{}}(A,B_{SW}) 
\geq \frac{c_{min}}{L}\log\left(\frac{L}{c_{min}}\right)SW(A,B_{SW}),
\end{align*}

Which entails:

$$\frac{SW(A,B_{\PAV{}})}{SW(A,B_{SW})} \geq \frac{\frac{c_{min}}{L}\log\left(\frac{L}{c_{min}}\right)SW(A,B_{SW})}{SW(A,B_{SW})} = \frac{c_{min}}{L}\log\left(\frac{L}{c_{min}}\right),$$
as required.
\end{proof}

\begin{proof}[Proof of Lemma~\ref{lemma:pav}]

We will use the following notation: Let $B_i = A(i)\cap B$ be the projects in  bundle $B$ that voter $i$ approves, $V(B)$ is the set of  voters with $|B_i|>1$ and $V1(B)$ are all voters with $|B_i|=1$. 

The harmonic sum $H(k)=1 + \frac{1}{2} + \ldots + \frac{1}{k}$ is at least $\log(k)$ (in base $e$), therefore, for any bundle $B$ and approval profile $A$,

\rmr{you should say somewhere that all logs are in base $e$. Note caveat after Eq.~\eqref{eq:temp1}}
\begin{equation}\label{eq:pav_lim}
\begin{split}
SC_{\PAV{}}(A,B) & = \sum_{i\in V}H(B_i) = \sum_{i\in V(B)}H(B_i) + |V1(B)| \\
 & \geq  \sum_{i\in V(B)}\log(|B_i|) + |V1(B)| 
  = \log(\prod_{i\in V(B)}|B_i|) + |V1(B)|
\end{split}
\end{equation}

And the welfare of $B$ is:

\begin{equation}\label{eq:pav_sw_eq}
SW(A,B) = \sum_{i\in V}|B_i| = \sum_{i\in V(B)}|B_i| + |V1(B)|
\end{equation}

From Eq.~\eqref{eq:pav_lim} and Eq.~\eqref{eq:pav_sw_eq} we have:

\begin{equation}
\frac{SC_{\PAV{}}(A,B)}{SW(A,B)}  \geq  \frac{\log(\prod_{i\in V(B)}|B_i|) + |V1(B)|}{\sum_{i\in V(B)}|B_i| + |V1(B)|}   \geq  \frac{\log(\prod_{i\in V(B)}|B_i|)}{\sum_{i\in V(B)}|B_i|}\label{eq:BB}
\end{equation}
We now want to find a lower bound on the right hand side of the last equation. Note that the lower bound does not depend on the actual $B_i$'s but only on their size. 

Intuitively, we want to show that the lowest value is obtained when the sets sizes' are most unbalanced---essentially when there is only one nonempty set.

For this, we solve the following optimization problem instead, which is a relaxation of the above problem. Let $T= |V(B)|$ and $Q=\sum_{i \in V(B)} |B_i|$. Define the convex set $$\mathcal{C} = \{(q_1,\ldots,q_T) \text{ s.t. } \forall i\, q_i \ge 2 \text{ and } \sum_{i=1}^T q_i = Q\}.$$
Using the notation, for any bundle $B$, the right hand side in Eq.~\eqref{eq:BB} is lower bounded by \rmr{min -> inf?}

\begin{align}\underset{\mathcal {C}}{\inf} \frac{\log(\prod_{i\in T}q_i)}{\sum_{i\in T}q_i} = \underset{\mathcal {C}}{\inf}  &\frac{\log(\prod_{i\in T}q_i)}{Q}.\label{eq:temp1}
\end{align}

The product of the $q_i$'s is minimal when the distribution of $q_i$'s is the most unbalanced. By setting the minimal value $q_i=2$ for all $i>1$, and $q_1=Q-2(T-1)$, we get

$$ \eqref{eq:temp1} \ge \inf_{T}\frac{\log (q_1 2^{T-1})}{Q} = \inf_T\frac{(T-1)\log 2+\log (Q-2(T-1))}{Q}.$$

Taking the derivative w.r.t. $T$, we get that this is a convex function with a maximum at $T=\frac{Q}{2}+1-\frac{1}{\log 2}>\frac{Q}{2}-1$. However $Q\geq q_1-2(T-1)\geq 2-2(T-1)=2T$, i.e. $T \leq \frac{Q}{2}$ so the only possible integer solutions are $T=1$ and $T=\frac{Q}{2}$, which map to $\frac{\log Q}{Q}$ and $\frac{\log 2}{2}$, respectively.  

Hence the minimum is obtained at  $T=1$, which means that $q_1 = Q$ (the solution of the relaxed problem is also a valid solution of our original problem). Back to bundle problem, $T=1$ entails $|V(B)|=1$ and $|B_1|=Q$. Plugging this back into Eq.~\eqref{eq:BB},

\begin{equation}\label{eq:pav_minimal_B}
\underset{B}{\arg\min} \frac{SC_{\PAV{}}(A,B)}{SW(A,B)} \geq \frac{\log(|B_1|)}{|B_1|} \geq \frac{\log(|B|)}{|B|}
\end{equation}

Since the  size of any feasible bundle is at most $\lfloor\frac{L}{c_{min}}\rfloor\leq\frac{L}{c_{min}}$, we get from Eq.~\eqref{eq:pav_minimal_B}:

\begin{equation}
\underset{B}{\arg\min} \frac{SC_{\PAV{}}(A,B)}{SW(A,B)} \geq \frac{c_{min}}{L}\log(\frac{L}{c_{min}})
\end{equation}


This completes the proof of Lemma \ref{lemma:pav}.
\end{proof}




\begin{figure}[t]
\begin{center}
\includegraphics[width=8cm]{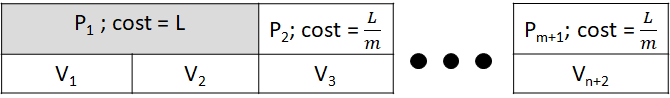}
\caption{PB instance for Prop.~\ref{prop:seqpav} showing projects $p_1, p_2, \dots p_{m+1}$ and voters $v_1, v_2, \ldots, V_{n+1}$. Each voter appears below its approved projects. In this example the  \sPAV{}  selects project $p_1$.\rmr{I would put a bit more energy in the figure... add the project costs and/or show the winning bundles}
}\label{fig:sPAV}
\end{center}\vspace{-5mm}
\end{figure}

\begin{proposition}\label{prop:seqpav}
$$\forall N,L,c_{min}:\
\frac{c_{min}}{NL} \leq K^{\sPAV{}}_{SW}(N,L,c_{min}) \leq \frac{2}{N}$$
$$\forall N:\ \frac{1}{N}\leq K^{\sPAV{}}_{RP}(N)\leq \frac{2}{N}$$
\end{proposition}

\begin{proof}
Consider the PB instance presented in Figure~\ref{fig:sPAV} with a budget of $L$. There are $N=n+2$ voters $\{v_1,\ldots,v_{n+2}\}$ and $M=m+1$ projects $\{m_1,\ldots,m_{m+1}\}$.

The first 2 voters approve project $p_1$ that costs $L$ and the rest of the voters approve one project each, at a cost of $\frac{L}{m}$.

At its first iteration, \sPAV{} chooses $p_1$, adding 2 to the score, while the addition of any other project adds 1. Therefore, \sPAV{} will fund $B_{\sPAV{}} = \{p_1\}$ and stop,\rmr{what is $p_x$??} having an outcome with welfare of 2. The bundle $B_{SW}=B_{RP}=\{p_2,\ldots,p_{m+1}\}$ maximizes both  welfare and representation with value  $n$. Putting things together we get that

\begin{equation}
\frac{SW(A,B_{\sPAV{}})}{SW(A,B_{SW})} = \frac{RP(A,B_{\sPAV{}})}{RP(A,B_{\CC{}})}= \frac{2}{n} = \frac{2}{(N-2)} \geq \frac{2}{N}
\end{equation}

As for the lower bound, any voting rule will fund at least one project and for any instance there can be at most $\lfloor\frac{L}{c_{min}}\rfloor$ projects funded that all voters want, therefore:

\begin{equation}
\frac{SW(A,B_{\sPAV{}})}{SW(A,B_{SW})} \geq \frac{1}{N\lfloor\frac{L}{c_{min}}\rfloor} \geq \frac{c_{min}}{NL}
\end{equation}

\begin{equation}
\frac{RP(A,B_{\sPAV{}})}{RP(A,B_{\CC{}})} \geq \frac{1}{N}
\end{equation}

\end{proof}

The rest of the results for \AV{}, \CC{}, \PAV{} and \sPAV{} can be seen in Table~\ref{tab:guarantees_util} and Table~\ref{tab:guarantees_rep}. Their proofs are left for the appendix, where the upper guarantees follow the same idea as in Proposition~\ref{prop:seqpav}, and the rest take into advantage the voting rule properties as done for Proposition~\ref{prop:pav_util_low} (Propositions~\ref{prop:av_high},~\ref{prop:av_low},~\ref{prop:cc_high} follow the same outlines as~\citet{lackner2020utilitarian}).

\subsection{EJR Voting Rule Guarantees}
In this section we will present the utilitarian and representation guarantees for the family of EJR voting rules.







\begin{proposition}\label{prop:ejr_util_low}
Let $R$ be a voting rule that satisfies the EJR property. Then the utilitarian guarantee satisfies
$$\forall N,L,c_{min}:\ K^{R}_{SW}(N,L,c_{min}) \geq \frac{c_{min}}{NL}\lfloor\frac{L}{c_{max}}\rfloor.$$
\end{proposition}

\begin{proof}

To avoid trivialities, we consider rules that exhaust the entire budget (EJR does not require that).
First let us lower bound the SW of an EJR rule R with respect to some PB instance E. Let $T \subseteq P$ be the largest set of projects that is $T$-cohesive with respect to $E$. Let $B$ be a bundle in the outcome of $R$ and $B \subseteq B'$ its extension to consume the remaining budget. From Definitions \ref{def:cohis} and \ref{def:EJR} it readily follows that any bundle $B$ in the outcome of R satisfies $SW(A,B) \ge \lceil |T|\frac{N}{L} \rceil$. 

In the worse case $T=\emptyset$, namely R is EJR in an empty way. Since we assume that all the budget is consumed, then $B'$ (perhaps even $B$) contains at least $\lfloor\frac{L}{c_{max}}\rfloor$ projects (otherwise the budget is not consumed). Therefore $SW(A,B') \ge \lfloor\frac{L}{c_{max}}\rfloor$.

As for the bundle $B_{SW}$ that maximizes the social welfare, there are at most $\lfloor\frac{L}{c_{min}}\rfloor$ projects possible to fund, each one of them is supported by at most all $N$ voters. This means that
$$SW(\AV{})\leq N\lfloor\frac{L}{c_{min}}\rfloor$$
Putting everything together we get that
$$K^{R}_{SW} \ge \frac{SW(A,B')}{SW(A,B_{SW})} \geq \frac{\lfloor\frac{L}{c_{max}}\rfloor}{N\lfloor\frac{L}{c_{min}}\rfloor}\geq\frac{c_{min}}{NL}\lfloor\frac{L}{c_{max}}\rfloor$$

\vspace{-5mm}
\end{proof}

The rest of the results for EJR voting rules in PB can be seen in Table~\ref{tab:guarantees_util} and Table~\ref{tab:guarantees_rep}. The proofs for those results are in Prop.~\ref{prop:ejr_rep} and Prop.~\ref{prop:ejr_util_high}


In order to have complete comparison of the PB results to multi-winner, we will also find the guarantees for EJR voting rules in multi-winner.
Notice that in  multi-winner context, $L\in\mathbb{N^+}$ and tell how many projects should be selected.

\begin{proposition}\label{prop:mw__high_in_text}
Let $R$ be a voting rule that satisfies the EJR property. Then the utilitarian and representation guarantees satisfies \rmr{you say utilitarian bu show for both}
$$\forall N,L,\  s.t.\  N\geq 2L:\ K^{R}_{SW}(N, L, c_{min}=1, c_{max}=1)\leq \frac{1}{N-L}$$ 
$$\forall N,L,\  s.t.\  N\geq 2L:\ K^{R}_{RP}(N,L, c_{min}=1, c_{max}=1) \leq \frac{1}{N-L}$$.
\end{proposition}\vspace{-3mm}
It is worth noting that while the bound in Prop.~\ref{prop:mw__high_in_text} (proof in Prop.~\ref{prop:mw__high}) can slightly improve on our bound for general PB problems (Prop.~\ref{prop:ejr_util_high}), asymptotically Prop.~\ref{prop:ejr_util_high} provides a tighter bound, that also shows some PB instances are worse (in terms of welfare) than any MW instance.    
\begin{proposition}\label{prop:mw_rep_low}
Let $R$ be a voting rule that satisfies the EJR property. Then the representation guarantee satisfies

$$\forall N:\ K^{R}_{RP}(N, c_{min}=1, c_{max}=1)\geq \frac{1}{N}$$
\end{proposition}

This is the trivial guarantee, as the optimal outcome represents at most all voters, and the voting rule outcome represents at least a single voter.

\begin{proposition}\label{prop:mw_util_low}
Let $R$ be a voting rule that satisfies the EJR property. Then the utilitarian guarantee satisfies
$$\forall N:\ K^{R}_{SW}(N, c_{min}=1, c_{max}=1) \geq \frac{1}{N}.$$
\end{proposition}
\noindent Proposition \ref{prop:mw_util_low} is proven the same way as Proposition~\ref{prop:ejr_util_low}.

\vspace{1mm}

We end this section with three conclusions that we draw from Table~\ref{tab:guarantees_util} and Table~\ref{tab:guarantees_rep}:
\begin{itemize}
    \item  The guarantees for \CC{} and \AV{} are the same for the multi-winner and PB settings (up to a $c_{max}$ factor in the representation lower bound of \AV{}). The case for \PAV{} and \PAV{} is very different. The PB guarantees are an order of magnitude lower, and for \PAV{} they also depend on $N$;  no multi-winner guarantee depends on $N$. This caused because they "ignore" the projects cost, however, less significant for \PAV{}, as it take the cost into account indirectly when solving the optimization problem.
\item Our results induce a nearly-strict order over the voting rules:
\begin{align*}
    \text{welfare:~~~~}&\AV{} \gg \PAV{} \gg \CC{} \gg EJR \gg \mbox{\PAV{}}\\
    \text{representation:~~~~}&\CC{} \gg \PAV{} \gg \AV{} \gg \mbox{\PAV{}}, EJR
\end{align*}
These two rankings are similar to the ones obtained in the multi-winner setting, except for \PAV{} which drops to the bottom when introducing costs (the PB setting). 
\item Most voting rules' guarantees depend on the budget and projects cost, while EJR voting rules guarantees depend on the number of voters and projects cost. This means that as the number of voters grows, the ``cost of proportionality" may be rising as well.
\end{itemize}


\section{Experimental Evaluation}\label{sec:exp}
 In this section we examine the performance of the rules in practice on real world and synthetic  data, beyond the worst-case scenario.

For every dataset and every voting rule we calculate the utilitarian ratio and representation ratio for all the instances in that dataset. We report the average and standard error. Tables \ref{tab:poland},\ref{tab:euc} and \ref{tab:party} report the results for the three datasets. 

{\bf Poland} A dataset that was taken from Pabulib.org~\cite{stolicki2020pabulib}, a library of PB instances available to the research community.
We looked at 130 instances that took place in different districts of  Warsaw, Poland, in the years 2017--2021. Each instance included between 50-10,000 voters (2,982 on average) and between 20-100 projects (36 on average). 

\newcommand{\newpar}[1]{\smallskip\noindent\textbf{#1}}
\newpar{Euclidean} This dataset consists of 1,000 synthetic PB instances, each containing  1,000 voters, 100 projects and a budget of $L=10^5$.
In a  city, most of the population lives near the city center, and so the location of a project is more likely to be there as well.
To this end, for every instance, the locations of the voters and projects are generated  randomly according to a 2-dimensional euclidean model~\cite{skowron2020participatory, elkind2017multiwinner, talmon2019framework}.
Each voter $v$ and project $p$ are given some location $\ell_v,\ell_p$ in the unit square $[0,1]\times [0,1]$, according to the normal distribution
\[
\mu= \begin{pmatrix}
  0.5 \\ 0.5
\end{pmatrix},
\Sigma=
\begin{pmatrix}
  0.2^2 & 0 \\ 0 & 0.2^2
\end{pmatrix}
\]

The costs of the projects are parameterized with two values $c_{min} \in [100,500]$, $c_{avg} \in [10^4,2\cdot 10^4]$, the minimum project cost and the average one, are both chosen uniformly at random. 

The project costs are chosen according to the following procedure. For every project $p_i$ choose $c_i$ from the exponential distribution with $\lambda = c_{avg}-c_{min}$ and the cost of $p_i$ will be $c_{min}+c_i$. This simulates a scenario with many cheap projects, and  a few expensive ones.

To create the approval profile  for each voter $i$, a number $a_i$ is chosen from the normal distribution with $\mu=10,\sigma=3$, and the set of projects $A(i)$ approved by voter $i$   consists of the $\max(a_i,1)$  closest  projects   to the location of voter $i$.


{\bf Party-list} This is also a synthetic dataset containing 1,000 party-list~\citep{lackner2020utilitarian} PB instances that satisfy the following condition: every pair of voters $i,j$, either approve the same list of projects, $A(i)=A(j)$, or don't approve any mutual one, $A(i) \cap A(j) = \emptyset$.
Each instance includes 200 voters which are split uniformly at random   into groups of sizes 5 to 20. 
Each group of voters  approves uniformly at random between 10 to 30 different projects. The cost of the projects is linear in the group size, such that the more voters a group has, the higher the cost of the projects they approve.

Due to the criteria above, PB instances in the party-list dataset contain large parties that tend to  approve expensive projects. Funding these projects will contribute significantly to the overall SW, but will consume large part of the budget, risking that small parties will not be represented, thus violating the EJR property.

\subsection{Voting Rules}
To apply the voting rules described in Section~\ref{sec:vr} to the PB datasets, we extended the python framework used by \citet{abcvoting}, originally designed  to find committees in the multi-winner setting. The \AV{}, \CC{} and \PAV{} were solved using linear programming with the  Gurobi solver~\cite{gurobi}.

For the sake of efficiency, instead of breaking ties for the worst-case ratio, we broke them at random. This allowed to test the voting rules on larger instances in reasonable time. On several random instances that we sampled, we did break ties for the worst-case and also for the best-case, and noted no significant change in the final score compared to the random policy.

One exception is \CC{} in the party-list dataset. In this dataset, it is possible to get 100\% representation with only a small portion of the budget spent. This leads to a variety of \CC{}-optimal bundles, achieving a wide range of welfare scores. The Gurobi solver selected, for unclear reasons, only solutions with high social welfare. To compensate for that, in this case only, we took the worst-case SW solution, and reported this result in Table \ref{tab:party} (we added a penalty for SW in the objective function and then ran Gurobi).

As mentioned in Section~\ref{sec:vr}, none of the voting rules  satisfies the EJR property. We use Rule~X~\cite{peters2020proportional} (\RX{}) as a representative from the EJR family. We  refer the reader to the original paper for a description of this rule \cite{peters2020proportional}. 
One caveat is that \RX{} does not necessarily exhaust the entire budget, in contrast to the other rules we consider. To allow a fair comparison with the other rules we will also consider two extension  to Rule-X.  The first variant, \RXeps{}, is described in \citet{peters2020proportional}. The second variant, \RXPAV{},  applies \RX{} to the PB instance, and runs PAV on the remaining budget over the unfunded projects. The outcome of \RXPAV{} is defined as the  union  of the outcomes of the  \RX{} and \PAV{} rules.

\subsection{Results}


\begin{figure}[t]
\begin{center}
\includegraphics[width=8cm]{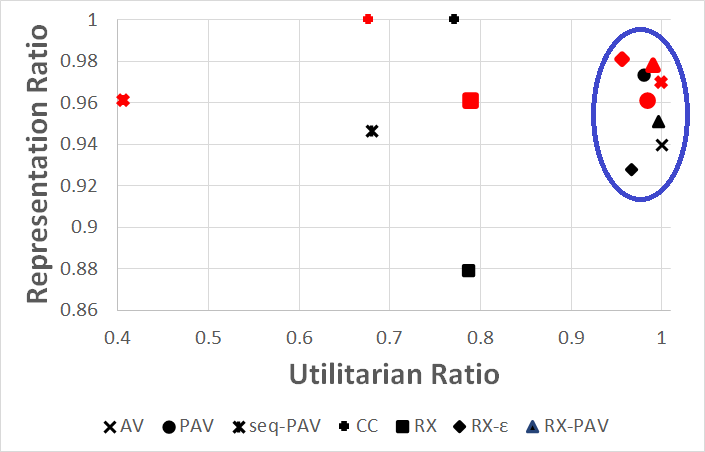}
\caption{Utilitarian ratio vs. Representation ratio results for the Poland (black) and Euclidean (bigger red).\rmr{impossible to read in BW print. Better use different shapes for rules and two colors (bright/dark) for the datasets}
}\label{fig:ratios_cluster}
\end{center}
\vspace{-5mm}
\end{figure}

Figure~\ref{fig:ratios_cluster} shows the trade-off between welfare and representation 
when applying the voting rules on Poland and Euclidean datasets (the party-list dataset was omitted in this figure as all voting rules except \AV{} gave an outcome with 100\% representation).
We see a cluster of voting rules (marked with blue circle) that includes \AV{}, \PAV{}, \RXeps{} and \RXPAV{}, achieving the best trade-off between welfare and representation. On the other, the outlier \sPAV{} achieves low welfare ratio in the Poland dataset and a very low welfare ratio in the Euclidean dataset.


Tables~\ref{tab:poland},\ref{tab:euc} and~\ref{tab:party} provide a finer level of granularity of the results, by displaying  the average ratios and percentages (with standard error\rmr{this is standard error! not standard deviation. It describes your confidence about the mean, not the population}) of PB instances that satisfy EJR for the Poland, Euclidean and Party-list datasets, respectively.  
As can be seen in the three tables, the ratios are similar across datasets. Specifically, both \PAV{} and \RXPAV{} succeed in achieving high utilitarian and representation ratios for all datasets. Another noticeable result, is the fact that \sPAV{} achieves quite poor results in both ratios,  which is in-line with the ranking that we presented at the end of Section \ref{sec:guar}. In addition, \sPAV{} exhibits large variance in all three  datasets (compared to the other voting rules), which further emphasizes that the \sPAV{} rule is unstable.

The results from the Poland  dataset in Table~\ref{tab:poland} demonstrate the disadvantage of using an EJR voting rule that does not guarantee to exhaust the budget.  For this dataset, \RX{} achieved poorer results for both welfare and representation compared to all other voting rules. In contrast, \RXeps{} and \RXPAV{}, which are similar to  \RX{}, but make sure to exhaust the entire budget, succeed in gaining a significant improvement in both measurements. This result emphasizes the benefit  of using the entire budget, even if the EJR requirement is satisfied before exhausting the budget.

Lastly, looking at percentage of instances where the chosen bundle satisfied EJR, we can see that for both the Poland and Euclidean datasets (Tables~\ref{tab:poland},\ref{tab:euc}), 
almost all rules succeed in getting an outcome which satisfies EJR. This  result is interesting, in that  even though a voting rule is not guaranteed to always produce a solution that satisfies EJR, this is often the case. 
This phenomenon may be explained by the fact that there are $T$-cohesive groups only for small sets of projects $T$ in those datasets, in other words voters are entitled to only a few projects. This makes the EJR requirement easier to satisfy. This property is satisfied for example in the Euclidean dataset where all projects have roughly the same, small, number of users that approve them.

In contrast to the above, in the party-list dataset (Table~\ref{tab:party}) most outcomes do not satisfy EJR (unless of course when the rule is part of the EJR family). \AV{}, \sPAV{} and \CC{} did not satisfy EJR in any instance.  \PAV{} satisfies EJR for about 80\% of the instances and provides a good utilitarian and representation ratios. However, the \RX{}-variants, \RXeps{} and \RXPAV{}, achieve the same ratios, but, also satisfy the EJR property,  making them  more attractive than \PAV{}.

The poor EJR percentages of \AV{}, \sPAV{} and \CC{} for the party-list dataset can be explained by observing the following. First, this dataset  forces large cohesive voter groups, making it more difficult for voting rules to satisfy EJR. Second, the PB instances contain projects that give high welfare or representation, but are also more expensive. The voting rules \AV{}, \sPAV{}, \PAV{} and \CC{} ignore cost, and thus by choosing such expensive projects, the budget is eaten fast, and small groups, with cheap projects, are not funded, violating the EJR property.

Finally, we consider the relationship between run-time and the guarantees. The rules \PAV{}, \RX{} and \RXeps{} run in polynomial  time  in  the number of projects and voters, while the rest of the voting rules take exponential time.
While \RXeps{} provides a good trade-off between all measurements,  in practice the run-time of this voting rule is an order of magnitude more time consuming than \emph{all} of the other voting rules, across all three datasets. 
Therefore, using the exponential-time \PAV{} or \RXPAV{} might be preferred in relatively small instances, where in practice we observed fast termination.

If we were to prepare a recommendation list for which rule to use when, taking into consideration the performance of the rules with respect to both guarantees and the run-time, then
the \PAV{} rule offers a good compromise across the board but it is computationally feasible only on small  instances. The \RXPAV{} rule achieves similar results to \PAV{} in addition to satisfying EJR, and it exhibits shorter run-time, since the exponential part of \RXPAV{} (the \PAV{} part) is applied to the remaining budget and unfunded projects (which is a much smaller instance). Hence \RXPAV{} is suitable both for small and medium instances.
Lastly, the rule \RXeps{} provides lower utilitarian and representation guarantees compared to \PAV{} and \RXPAV{} in addition to satisfying EJR, but it runs in polynomial time, making  it the rule of choice  for large instances.

\begin{table}[t]
  \begin{center}
    \begin{tabular}{l|c|c|c}
      & Utilitarian ratio  & Representation ratio & EJR\%\\
      \hline
      \AV{} & $1 \pm 0$ &	$0.9401 \pm 0.054$ &	100\\
      \PAV{} & $0.9802  \pm 0.025$	& $0.9735 \pm 0.026$ &	100\\
      \sPAV{} & $0.6801 \pm 0.18$ & $0.9468 \pm 0.073$ & 94.6\\
      \CC{} & $0.7713 \pm 0.12$ &	$1 \pm 0$ &	100\\
      \RX{} & $0.7868 \pm 0.065$ &	$0.8793 \pm 0.08$ &	100\\
      \RXeps{}  & $0.9666 \pm 0.025$ &	$0.9278 \pm 0.063$ &	100\\
      \RXPAV{} & $0.9966 \pm 0.006$ &	$0.9508 \pm 0.044$ &	100\\
    \end{tabular}
    \caption{\label{tab:poland}Ratios and percentage of instances which satisfy EJR  (Poland dataset)}
  \end{center}
  \vspace{-8mm}
\end{table}

\begin{table}[t]
  \begin{center}
    \begin{tabular}{l|c|c|c}
      & Utilitarian ratio  & Representation ratio & EJR\%\\
      \hline
      \AV{} & $1 \pm 0$ &	$0.9696 \pm 0.021$ &	100\\
      \PAV{} & $0.9845 \pm 0.011$ &	$0.9866 \pm 0.008$ &	100\\
      \sPAV{} & $0.4063 \pm 0.095$ &	$0.8423 \pm 0.089$ &	64.7\\
      \CC{} & $0.6768 \pm 0.075$ &	$1 \pm 0$ &	99.9\\
      \RX{} & $0.7890 \pm 0.037$ &	$0.9609 \pm 0.024$ &	100\\
      \RXeps{}  & $0.9569 \pm 0.017$ &	$0.9808 \pm 0.012$ &	100\\
      \RXPAV{} & $0.9905 \pm 0.007$ &	$0.9780 \pm 0.014$ &	100\\
    \end{tabular}
    \caption{\label{tab:euc}Ratios and percentage of instances which satisfy EJR for euclidean dataset}
  \end{center}
  \vspace{-8mm}
\end{table}

\begin{table}[t]
  \begin{center}
    \begin{tabular}{l|c|c|c}
      & Utilitarian ratio  & Representation ratio & EJR\%\\
      \hline
      \AV{} & $1 \pm 0$ &	$0.6628 \pm 0.044$ &	0\\
      \PAV{} & $0.8459 \pm 0.023$ &	$1 \pm 0$ &	79.2\\
      \sPAV{} & $0.7348 \pm 0.053$ &	$1 \pm 0$ &	0\\
      \CC{} & $0.6558 \pm 0.056$ &	$1 \pm 0$ &	0\\
      \RX{} & $0.8125 \pm 0.024$ &	$1 \pm 0$ &	100\\
      \RXeps{}  & $0.8579 \pm 0.023$ &	$1 \pm 0$ &	100\\
      \RXPAV{} & $0.8536 \pm 0.022$ &	$1 \pm 0$ &	100\\
    \end{tabular}
    \caption{\label{tab:party}Ratios and percentage of instances which satisfy EJR for party-list dataset}
  \end{center}
  \vspace{-10mm}
\end{table}


\vspace{-3mm}
\section{Conclusions and Future Work}\label{sec:conclusion}

We  presented a theoretical and empirical investigation of  the trade-off  between welfare and representation for different voting rules for participatory budgeting. 
From the theoretical perspective, we analyzed the worst-case guarantees of common voting rules from the literature. We show that when introducing costs to  projects, these guarantees do not generalize to the PB setting, with some rules (e.g., \sPAV{}) exhibiting  significantly lower guarantees than the multi-winner setting. 
From the empirical perspective, we show that some  proportional  voting rules  (namely \RXPAV{} and \RXeps{}) are able to  achieve high social  welfare and representation on real PB instances, in contrast to their theoretical guarantees.

Taking into consideration the trade-off between welfare and representation,   we concluded the  \PAV{} rule to be  the clear winner from both theory and practice perspectives, however, it is not proportional, and exhibits worst case exponential running time. This  led us to analyze  two variants of \RX{} (\RXPAV{} and \RXeps{}) that exhibit proportionality and provide similar results to \PAV{} in practice, despite their   lower theoretical guarantees (expressing the cost of proportionality). 
Specifically, we claimed that   \RXPAV{} is suitable for solving medium sized PB instances, while \RXeps{} is suitable for large PB instances as a result of their run-time.
 Our results provide a deeper understanding of the trade-offs between welfare and representation for  voting rules in Participatory Budgeting and can lead to more efficient outcomes that will satisfy the citizens.

There are several directions that are interesting to explore in  future work. First, extending our analysis to consider additional voting rules from the  literature, and considering more families of rules (e.g. voting rules with constrain on minority representation\rmr{would be nice to give an example of a different `class' of rules}).
Second, 
 finding  sufficient conditions on PB instances such that a voting rule
 would satisfy EJR. 
This could be a way to reconcile proportionality with other requirements, albeit for a restricted class of participatory budgeting problems.

\section{Acknowledgements}
This work was supported in part by Israeli Science Foundation (ISF) Grant No. 773/16.

Thanks to Jannik Peters who found a mistake in the EJR upper bounds in the original version.






\bibliographystyle{ACM-Reference-Format} 
\bibliography{ref.bib}


\clearpage
\appendix

\section{Proof details from Section~\ref{sec:guar}}

\begin{proposition}\label{prop:av_low}
$\forall L,c_{min},c_{max}:\  K^{\AV{}}_{RP}(L,c_{min},c_{max}) \geq \frac{c_{min}^2}{Lc_{max}}$

\end{proposition}

\begin{proof}
Lets note by $p_{max}$ the project which was approved by most voters. A bundle chosen by \AV{}, must either contain $p_{max}$ or contain several cheaper projects $\{P_1,\ldots,P_x\}$ such that their total \AV{} score will be at least $N(p_{max}):=N_{max}$. We will notice that in the worst case, all projects will represent the same voters thus achieving the lowest representation. Since those projects must have certain welfare, the bigger $x$ is the less voters the projects must represent. The largest number of projects that we can have happen when $p_{max}$ is the most expensive project and cost $c_{max}$ and all other projects are the cheapest and cost $c_{min}$, resulting with $X=\lfloor\frac{c_{max}}{c_{min}}\rfloor$.

In the worst case all of those projects represent the same voters, i.e. they will represent at least $\lceil\frac{N_{max}}{\lfloor\frac{c_{max}}{c_{min}}\rfloor}\rceil \geq \frac{N_{max}c_{min}}{c_{max}}$.

In addition, since there can be at most $\lfloor\frac{L}{c_{min}}\rfloor$ projects funded, we get for any chosen bundle B:
$$RP(A, B)\leq\lfloor\frac{L}{c_{min}}\rfloor N_{max}\leq\frac{LN_{max}}{c_{min}}$$

This results with:
$$\frac{RP(A,B_{\AV{}})}{RP(A,B_{RP})}\geq \frac{N_{max}c_{min}}{c_{max}} / \frac{LN_{max}}{c_{min}} = \frac{c_{min}^2}{Lc_{max}}$$
\end{proof}

\begin{proposition}\label{prop:av_high}
$\forall L,c_{min}:\  K^{\AV{}}_{RP}(L,c_{min})\leq \frac{c_{min}}{L}$

\end{proposition}

\begin{proof} 
Let $M=m^2$ for some $m \in \mathbb{N}$, and let $x \in \mathbb{N}_{+}$. Consider a PB instance $E=(A,\frac{L}{m},L)$ with $M$ projects and $N=mx + 1$ voters. We split the set of voters into $m$ groups (the first group includes $x+1$ voters and the rest $x$ voters), each group approves $m$ projects, without intersection between groups. Each project costs $\frac{L}{m}$, i.e., it is possible to fund at most $m$ projects.

The bundle $B_0$ which funds all projects of the first group belongs to $R_{\AV{}}(E)$; its social welfare is maximal at $SW(A,B_0)=m(x+1)$ and it representation score is $RP(A,B_0)=x+1$. In contrast, the optimal representation bundle $B_{opt}$ contains one project from each group, yielding $RP(A,B_{opt})=mx + 1$. Putting it together,
\begin{equation}\label{eq:av}
\frac{RP(A,B_0)}{RP(A,B_{opt})} = \frac{x+1}{mx + 1} \leq \frac{1}{m} + \frac{1}{mx} =\frac{c_{min}}{L} \left(1 + \frac{1}{x}\right) \le \frac{2c_{min}}{L}. 
\end{equation}


\end{proof}

\begin{proposition}\label{prop:cc_low}
$\forall L,c_{min}:\  K^{\CC{}}_{SW}(L,c_{min})\geq \frac{c_{min}}{L}$
\end{proposition}

\begin{proof}

The proof follows from the following general observation. Let $E$ be a PB instance and  $B_{SW}$ a bundle that maximizes the social welfare, with value $s$. Since the size of $B_{SW}$ is at most $L/c_{min}$, it implies that each voter contributes to the welfare at most $\frac{L}{c_{min}}$. Therefore, the number of voters $i$ such that $|A(i) \cap B_{SW}| > 0$ is at least $s \cdot \frac{c_{min}}{L}$.

The outcome $B_{\CC{}}$ of \CC{} maximizes the number of $i$'s such that $|A(i) \cap B_{\CC{}}| > 0$, and in particular it is larger than the number of $i$'s such that $|A(i) \cap B_{SW}| > 0$. The latter is at least $s \cdot \frac{c_{min}}{L}$. Therefore,
$$\frac{SW(A,B_{\CC{}})}{s} \ge \frac{RP(A,B_{\CC{}})}{s} \ge \frac{c_{min}}{L}.$$

\end{proof}

\begin{proposition}\label{prop:cc_high}
$\forall N,L,c_{min}:\ K^{\CC{}}_{SW}(L,c_{min})\leq \frac{4c_{min}}{L}$
\end{proposition}

\begin{proof}
Consider the following PB instances: $V=\{v_1,\ldots,v_{2n+1}\}$, $P=\{p_1,\ldots,p_{2n+1}\}$, and the budget is L. The cost of projects $\{p_1,\ldots,p_{n}\}$ is $\frac{L}{n}$ and approved by all voters $\{v_1,\ldots,v_{n}\}$; the cost of projects $\{p_{n+1},\ldots,p_{2n+1}\}$ is $\frac{L}{n+1}$ and each one is approved by a single voter from $\{v_{n+1},\ldots,v_{2n+1}\}$.

Let  $k$  be the number of projects with cost $\frac{L}{n+1}$ that can be funded after funding one project with cost $\frac{L}{n}$, i.e. $k=\lfloor (L -\frac{L}{n})/\frac{L}{n+1}\rfloor$ = $\lfloor\frac{(n-1)(n+1)}{n}\rfloor=\lfloor\frac{n^2-1}{n}\rfloor$.

The bundle that maximizes SW is $B_{SW}=\{p_1,\ldots,p_n\}$ with $SW(A,B_{SW})=n^2$. An optimal bundle of \CC{} is given by   $B_{\CC{}}=\{p_1,p_{n+1},\ldots,p_{n+1+k}\}$. The social welfare score of $B_{\CC{}}$ is
$$SW(A,B_{\CC{}})=n+k= n+ \lfloor\frac{n^2-1}{n}\rfloor \leq n + \frac{n^2-1}{n}=2n - \frac{1}{n}.$$
Dividing by the optimal SW we get,
\begin{equation}\label{eq:cc_high}
\frac{SW(A,B_{\CC{}})}{SW(A,B_{SW})}\leq \frac{2n - \frac{1}{n}}{n^2} = \frac{2}{n} - \frac{1}{n^3} \leq \frac{2}{n} \le \frac{2c_{min}}{L-c_{min}}.
\end{equation}
The last inequality is from solving for $n$ in $c_{min} = \frac{L}{n+1}$.
If we assume that $c_{min} \leq L/2$ then Eq.~\eqref{eq:cc_high} is upper bounded by
$$\frac{SW(A,B_{\CC{}})}{SW(A,B_{SW})}\leq \frac{2c_{min}}{L - L/2} = \frac{4c_{min}}{L}.$$

If $c_{min} > L/2$ then there is only a single project that can be funded in which case finding optimal welfare or representation is a trivial task.
\end{proof}

\begin{proposition}\label{prop:pav_util_high}
$\forall N,L,c_{min}:\ K^{\PAV{}}_{SW}(L,c_{min})\leq \frac{c_{min}}{L}\left(\log\left(\frac{L}{c_{min}}\right) + 2\right)$
\end{proposition}

\begin{proof}
Consider the following PB instance with a project set $P = \{p_1,\ldots,p_{x+1}\}$, voters $V=\{v_1,\ldots,v_{\lfloor \log(x)\rfloor + 3}\}$, and budget of $L$. The first voter wants the first $x$ projects, each costs $c_{min} = \frac{L}{x}$ and the rest of the voters want the last project which costs $L$.

The largest \PAV{} score is obtained for $B_{\PAV{}}=\{p_{x+1}\}$ because:

\begin{align*}
SC_{\PAV{}}(A,\{p_1,\ldots,p_x\}) = \sum_{i=1}^x\frac{1}{i} \leq \log(x) + 1 \\ \leq \lfloor\log(x)\rfloor + 2 = SC_{\PAV{}}(A,\{p_{x+1}\})
\end{align*}

The bundle $B_{SW}=\{p1,\ldots,p_x\}$ maximizes the social welfare with $SW(A,B_{SW}) = x = \frac{L}{c_{min}}$. Together we get,
$$\frac{SW(A,B_{\PAV{}})}{SW(A,B_{SW})} = \frac{\lfloor \log\left(\frac{L}{c_{min}}\right)\rfloor + 2}{\frac{L}{c_{min}}} \leq \frac{c_{min}}{L}\left(\log\left(\frac{L}{c_{min}}\right) + 2\right)$$

\end{proof}

\begin{proposition}\label{prop:pav_rep_low}
$\forall L,c_{min}:\ K^{\PAV{}}_{RP}(L,c_{min}) \geq \frac{1}{2\log(\frac{L}{c_{min}})}$

\end{proposition}

\begin{proof}
Given a PB instance $E$, we will note by opt(E) the outcome with highest representation. For each voter $i$ it holds $|A(i)\cap B| \geq\frac{L}{c_{min}}$, therefore, $SC_{\PAV{}}(A,B) \leq SC_{\PAV{}}(\frac{L}{c_{min}})RP(A,B)\leq 2\log(\frac{L}{c_{min}})RP(A,B)$.

Which gives:

$$\frac{RP(A,\PAV{}(E))}{RP(A,opt(E))} \geq \frac{RP(A,\PAV{}(E))}{2\log(\frac{L}{c_{min}})RP(A,\PAV{}(E))} = \frac{1}{2\log(\frac{L}{c_{min}})}$$
\end{proof}

\begin{proposition}\label{prop:pav_rep_high}
$\forall N,L, s.t. \ N=\lfloor\log(L)\rfloor -1:\ K^{\PAV{}}_{RP}(L)\leq \frac{1}{\lfloor\log(L)\rfloor -1}$

\end{proposition}

\begin{proof}
Consider PB instance with budget of $L$, $N=\lfloor\log(L)\rfloor -1 $ voters and $M=\lfloor L\rfloor$ projects.

The first project is approved by all voters and cost $L$, and the rest of the projects approved only by the first voters, each with cost of 1.

There are two possible bundles to fund, either taking only the first project, or taking all other projects, giving \PAV{} score of:

\begin{align*}
SC_{\PAV{}}(A,\{p_2,\ldots,p_{\lfloor L\rfloor + 1}\}) = N = \lfloor\log(L)\rfloor -1  \leq \sum_{i=1}^L\frac{1}{i} = SC_{\PAV{}}(A,\{p_1\}) 
\end{align*}

Therefore, the \PAV{}-winning bundle is $B_{\PAV{}}=\{p_2,\ldots,p_{\lfloor L\rfloor + 1}\}$ with $RP(A,B_{\PAV{}}) = 1$ and the \CC{}-winning bundle is  $B_{\CC{}}=\{p_1\}$ with $RP(A,B_{\CC{}})=N=\lfloor\log(L)\rfloor -1$

This mean:

$$\frac{RP(A,B_{\PAV{}})}{RP(A,B_{\CC{}})} = \frac{1}{\lfloor\log(L)\rfloor -1}$$

\end{proof}

\begin{proposition}\label{prop:ejr_rep} Let $R$ be a voting rule that satisfies the EJR property. Then the representation guarantee satisfies \\
$$\forall N: \ K^R_{RP}(N) \geq \frac{1}{N}$$ 
$$\forall N, L,c_{min}, \ s.t.\  N\geq\frac{L}{c_{min}}: \ K^{R}_{RP}(N, L,c_{min}) \le \frac{1}{N-1}$$

\end{proposition}

\begin{proof}
Given a PB instance with $N$ voters and 3 projects. the first voter approve the first two  projects that cost $\frac{L}{2N}$ each, and the rest approve the third project that cost $\frac{N-1}{N}L + \epsilon$

According to EJR, the first votes should have at least 2 of the approve projects funded, resulting with not enough funds for the third project and only one voter represented. In contrast, the optimal, outcome will take the third project in addition to one of the first two projects, representing all voters. This results with:

$$\frac{RP(A,R(E))}{RP(A,cc(E))} = \frac{1}{N}$$.

The lower bound $K^{R}_{RP} \ge 1/N$ is immediate as any rule will fund at least 1 project and the optimal representation score is  at most $N$.

\end{proof}





\begin{proposition}\label{prop:ejr_util_high}
Let $R$ be a voting rule that satisfies the EJR property. Then the utilitarian guarantee satisfies
$$\forall N,L,c_{min}:\ K^{R}_{SW}(N,L,c_{min}) \leq \frac{4}{\sqrt{N}} - \frac{1}{N}.$$
\end{proposition}

\begin{proof}
Given a PB instance with $N$ voters and $2N - \lfloor\sqrt{N}\rfloor$ projects such that each one cost $\frac{L}{N}$. The first $\lfloor\sqrt{N}\rfloor$ voters approve projects $\{p_1,\ldots,p_N\}$ and the rest of the voters each approve a single project $\{p_{N+1},\ldots,p_{2N - \lfloor\sqrt{N}\rfloor}\}$.

From EJR definition, all singletons should be funded and the rest of the budget used for the rest of the projects. This gives:
$$SW(A,R(E)) = N - \sqrt{N} + \lfloor\sqrt{N}\rfloor\lfloor\sqrt{N}\rfloor \leq 2N - \sqrt{N}$$

While the optimal welfare is taking the first $N$ projects. This results with:

$$\frac{SW(A,R(E))}{SW(A,AV(E))} \leq \frac{2N - \sqrt{N}}{\lfloor\sqrt{N}\rfloor N} \leq \frac{4}{\sqrt{N}} - \frac{1}{N}$$
\end{proof}

\section{Rule~X description}\label{app:rx}

The voting rule Rule~X (\RX{}) recently introduced by \citet{peters2020proportional}. 
\RX{} is an iterative rule, which starts with ``allocating" each voter an equal share of the budget $\frac{L}{|V|}$, and  initialize an empty outcome $B=\varnothing$; then sequentially adds projects to $B$. At each step, in order to choose some project $p\in P\setminus B$, each voter needs to pay an amount that is proportional to her utility from the project, but no more than her remaining budget (note that with approval utilities this means only agents that approve the project pay). The  total payment should cover the cost of the project. 

Formally, let $b_i(t)$ be the amount of money that voter~$i$ is left with just before iteration~$t$.
 We say that some project $p\in P$, is $q$-affordable if $\exists q\in \mathbb{R}_+$ such that 
$$\sum_{i\in V}min(b_i(t),U_i(p)\cdot q)\geq cost(p)$$

Where
$U_i(p)= 
        \begin{cases}
            1             ,& \text{if } p\in A(i)\\
            0,             & \text{otherwise}
        \end{cases}$
is the utility of voter $i$ for project $p$.

If no candidate project is q-affordable for any $q$, Rule~X terminates and returns $B$. Otherwise it selects project $p^{(t)}\notin B$ that is $q$-affordable for a minimum $q$, where individual payments are given by $c_i(p^{(t)}):=\min\{b_i(t),u_i(p^{(t)})\cdot q\}$. Then we update the remaining budget as $b_i(t+1):=b_i(t)-c_i(p^{(t)})$.

The pseudocode for calculating the  qValue for a given project is shown in Algorithm~\ref{algo:qval}. 
The pseudocode for \RX{} is given in Algorithm~\ref{algo:rx}.

\begin{algorithm}[t]
\SetAlgoLined
\textbf{Input:} 
\begin{enumerate}
    \item project $p\in A$ \\ 
    \item $\forall i\in V,  U_i(p)$ \\
    
\end{enumerate}
\KwResult{q-value computation for project $p$}
\If {$\sum_{i\in v, U_i(p)>0}b_i(t)< cost(p)$}
 {$return \  \infty$}
 
 $current\_utility \leftarrow \sum_{i\in V}U_i(p)$
 
 $cost\_leftover \leftarrow cost(p)$
 
 $removed\_voters \leftarrow \varnothing$
 
 \While{True}{
    $current\_q \leftarrow cost\_leftover /    current\_utility$
    
    $voter\_removed \leftarrow False$
    
    \For{$i\in V\setminus removed\_voters$}{
      \If{$current\_q * U_i(p) > b_i(t)$}{
            $current\_utility \leftarrow current\_utility - U_i(p)$
            
            $cost\_leftover \leftarrow cost\_leftover - b_i(t)$
            
            $removed\_voters \leftarrow removed\_voters \cup \{i\}$
            
            $voter\_removed \leftarrow True$
       }
       
     }
     
     \If{$voter\_removed == False$}{
            $return \  (cost\_leftover / current\_utility)$
       }
 }

 \caption{qValue}\label{algo:qval}
\end{algorithm}

\begin{algorithm}
\SetAlgoLined
\textbf{Input:}
\begin{enumerate}
    \item $\forall i\in V, \forall g\in v_i, \forall p\in A,  u_{i,g}(p)$   \\
    \item Budget $L$ \\
    \item $\forall p\in A, cost(p)$
\end{enumerate}

\KwResult{Feasible bundle $B\subseteq A$}
 $B_0 \leftarrow \varnothing$

 $\forall i\in V: b_i(0) \leftarrow\frac{L}{|V|}$ 
  $t \leftarrow 1$

 \While{True}{
 
  $p^{(t)} \leftarrow \mathrm{argmin}_{p\in A\setminus B_{t-1}}[qValue(p, U_{[|A|]})]$
  
  \If{$qValue(p^{(t)}, U_{[|A|]})=\infty$}{
        $return \  B_{t-1}$ 
   }
   $B_{t} \leftarrow B_{t-1} \cup \{p^{(t)}\}$
   
   $\forall i\in V: b_i(t) \leftarrow \max\{0,b_i(t-1) -U_i(p^{(t)})\cdot q\}$
   
   $t \leftarrow t + 1$
   }
 \caption{Rule X}\label{algo:rx}
\end{algorithm}

\end{document}